\documentclass[conference]{IEEEtran}
\IEEEoverridecommandlockouts

\usepackage{graphicx}
\usepackage{algorithm}
\usepackage{stackengine}
\usepackage{cite}
\usepackage{color}
\usepackage{xcolor}
\usepackage{amsmath,amsthm,amssymb,amsfonts,bm}
\usepackage{textcomp}
\usepackage{multirow}
\usepackage{url}

\def\BibTeX{{\rm B\kern-.05em{\sc i\kern-.025em b}\kern-.08em
    T\kern-.1667em\lower.7ex\hbox{E}\kern-.125emX}}

\usepackage[noend]{algpseudocode}

\makeatletter
\newcommand*{\rom}[1]{\expandafter\@slowromancap\romannumeral #1@}

\newtheorem{theorem}{Theorem}

\def \v x{\bm x}

\def \v x{\bm X}

\renewcommand{\v}[1]{\ensuremath{\boldsymbol{#1}}}

\newcommand{\recht}[1]{\operatorname{#1}}

\begin{document}

\title{Privacy-Preserving Distributed Maximum Consensus Without Accuracy Loss}

\author{\IEEEauthorblockN{Wenrui Yu}
\IEEEauthorblockA{\textit{CISPA Helmholtz Center}\\
\textit{for Information Security}\\
Germany \\
wenrui.yu@cispa.de}
\and
\IEEEauthorblockN{Richard Heusdens}
\IEEEauthorblockA{\textit{Netherlands Defence Academy}\\
\textit{Delft University of Technology}\\
the Netherlands \\
r.heusdens@tudelft.nl}
\and
\IEEEauthorblockN{Jun Pang}
\IEEEauthorblockA{
\textit{University of Luxembourg}\\
Luxembourg\\
jun.pang@uni.lu}
\and
\IEEEauthorblockN{Qiongxiu Li}
\IEEEauthorblockA{
\textit{Aalborg University}\\
Denmark \\
qili@es.aau.dk}
}

%
\maketitle
\raggedbottom
\begin{sloppy}
\addtolength{\abovedisplayskip}{-1.0mm}
\addtolength{\belowdisplayskip}{-1.0mm}

\setlength{\lineskiplimit}{0pt}
\setlength{\lineskip}{0pt}
\setlength{\abovedisplayskip}{2pt}
\setlength{\belowdisplayskip}{2pt}
\setlength{\abovedisplayshortskip}{2pt}
\setlength{\belowdisplayshortskip}{2pt} 
\setlength{\belowcaptionskip}{-6pt}

\begin{abstract}
In distributed networks, calculating the maximum element is a fundamental task in data analysis, known as the distributed maximum consensus problem. However, the sensitive nature of the data involved makes privacy protection essential. Despite its importance, privacy in distributed maximum consensus has received limited attention in the literature. Traditional privacy-preserving methods typically add noise to updates, degrading the accuracy of the final result. To overcome these limitations, we propose a novel distributed optimization-based approach that preserves privacy without sacrificing accuracy. Our method introduces virtual nodes to form an augmented graph and leverages a carefully designed initialization process to ensure the privacy of honest participants, even when all their neighboring nodes are dishonest. Through a comprehensive information-theoretical analysis, we derive a sufficient condition to protect private data against both passive and eavesdropping adversaries. Extensive experiments validate the effectiveness of our approach, demonstrating that it not only preserves perfect privacy but also maintains accuracy, outperforming existing noise-based methods that typically suffer from accuracy loss.
\end{abstract}
\section{Introduction}
\label{sec:intro}

Consensus algorithms are designed to facilitate network-wide agreement through localized computations and the exchange of information among neighboring nodes. 
These algorithms represent a fundamental challenge in distributed optimization and have found widespread applications. 
Typical examples include averaging \cite{francca2020distributed,zhang2014asynchronous}, maximum/minimum \cite{deplano2023unified}, and median \cite{deplano2023unified} consensus.
However, since information sharing is an essential process in solving consensus problems, it raises severe privacy concerns. 

Common privacy preservation techniques in consensus problems include differential privacy (DP) \cite{huang2015differentially,nozari2018differentially,zhang2016dynamic,zhang2018improving,xiong2020privacy}, secure multi-party computation (SMPC) \cite{gupta2017privacy,li2019privacyA,tjell2020privacy,tjell2019privacy,xu2015secure,li2019privacyS,shoukry2016privacy,zhang2019admm},  subspace perturbation \cite{Jane2020ICASSP,Jane2020LS,li2020privacy} and variants of it \cite{jordan2024,li2022communication,li2023adaptive}. DP achieves a level of protection by adding zero-mean noise, thereby obfuscating the private data. However, this approach involves a tradeoff between utility and privacy; higher levels of noise lead to better privacy but result in reduced accuracy.
SMPC techniques, such as secret sharing \cite{Cramer2015}, often incur communication overhead due to the need to split and distribute the message for transmission.
Subspace perturbation, based on distributed optimizers such as the Alternating Direction Method of Multipliers (ADMM) \cite{boyd2011distributed}  or the Primal-Dual Method of Multipliers (PDMM) \cite{zhang2017distributed,heusdens2024distributed}, operates by introducing noise due to proper initialization of the optimization variables. Since algorithms are guaranteed to converge regardless of the initial conditions, the algorithm accuracy remains uncompromised. Consequently, it allows for privacy preservation while maintaining the integrity of the original data.

While average consensus has been extensively studied, the issue of privacy leakage in nonlinear consensus problems, such as maximum/minimum and median consensus, has received relatively little attention. The investigation can advance the understanding of Byzantine robustness in distributed systems, such as federated learning \cite{pillutla2022robust}. A few works have attempted to address this concern. Wang et al. \cite{wang2018differentially} directly adds Gaussian noise to private data before broadcasting it to the network, while Venkategowda et al.\cite{venkategowda2020privacy,gratton2021privacy} employs DP within the ADMM framework by adding noise to the primal variable. Unfortunately, these approaches introduce a tradeoff between accuracy and privacy. To overcome it, subspace perturbation \cite{Jane2020ICASSP,Jane2020LS,li2020privacy} has emerged as an attractive alternative, bypassing it through the proper initialization of auxiliary variables. However, two challenges arise when applying it to maximum consensus. Firstly, this technique was originally proposed for problems with equality constraints, it is unclear whether it works effectively for inequality constraints. Secondly, it guarantees the privacy of an honest node only if it has at least one honest neighbor, which may not always be practical.

In this paper, we propose a simple yet effective approach to address these challenges. Our method not only extends subspace perturbation to inequality constrained scenarios within maximum consensus but also introduces additional virtual nodes (referred to as dummy nodes) to form an augmented graph to ensure the privacy of honest nodes, even in the extreme case that all their neighboring nodes are dishonest. Our approach is grounded by information-theoretical analysis, from which we derive a sufficient condition to ensure (asymptotically) perfect privacy of honest nodes. To our knowledge, it is the first instance of a privacy-preserving maximum consensus algorithm that incurs no accuracy loss while being supported by rigorous information-theoretical analysis. Extensive experimental results consolidate the effectiveness of our approach. 
 
\section{Preliminaries}
\label{sec:pre}

\subsection{Problem formulation}
We model our network by a graphical model $\mathcal{G} = (\mathcal{V},\mathcal{E})$, where $\mathcal{V} = \{1,\ldots,n\}$ represents the set of nodes/participants in the network and $\mathcal{E} \subseteq \mathcal{V} \times \mathcal{V}$ represents the set of undirected edges indicating the connections between the nodes (communication links). For each node $i$ we denote its set of neighbors as $\mathcal{N}_i = \{j \in{\cal V} \,|\, (i,j) \in \mathcal{E}\}$ and its degree by $d_i = |\mathcal{N}_i|$.  Let $s_i\in \mathbb{R}$ denote the data\footnote{For simplicity, we assume $s_i$ is a scalar, but results can easily be generalized to the vector case by considering element-wise maximum operations.} in node $i\in{\cal V}$. The privacy-preserving maximum consensus problem is to find the maximum value $s_{\max} = \max \{s_i: i\in \mathcal{V}\}$ in the network without revealing the local data $s_i$. To do so, we formulate the optimization challenge as a linear programming (LP) problem \cite{venkategowda2020privacy}
\begin{equation}
\begin{array}{ll} \text{minimize} & {\displaystyle \sum_{i\in {\cal V}} x_i,} \\\rule[4mm]{0mm}{0mm}
\text{subject to} & x_i -x_j = 0, \quad (i,j)\in \cal  E,
\\
&x_i\geq s_i, \quad  i\in \mathcal{V}.
\end{array}
\label{eq:lp}
\end{equation}
When $x$ is updated iteratively, we write $x^{(t)}$ to indicate the update of $x$ at the $t$th iteration. When we consider $x$ as a realization of a random variable, the corresponding random variable will be denoted by $X$ (corresponding capital).

\subsection{A/PDMM with linear equality and inequality constraints}

 Following \cite{heusdens2024distributed}, we consider the minimization of a separable convex function subject to a set of inequality constraints by
\begin{equation}
\begin{array}{ll} \text{minimize} & \displaystyle\sum_{i \in \mathcal{V}} f_i\left(x_i\right), \\\rule[3mm]{0mm}{0mm}
\text{subject to} & A_{i j} x_i+A_{j i} x_j \preceq b_{i j}, \quad(i, j) \in \mathcal{E},
\end{array}
\label{eq:problem}
\end{equation}
where $f_i$ are convex, closed and proper (CCP) functions and $\preceq$ (generalized inequality) represents element-wise inequality.

To solve \eqref{eq:problem}, the update equations of the so-called inequality constraint primal-dual method of multipliers (IEQ-PDMM) \cite{heusdens2024distributed} for node $i\in \mathcal{V}$ are given by
\begin{align}
    \nonumber &x_i^{(t+1)}=\arg \min _{x_i}\bigg(f_i(x)\\ 
    \nonumber & ~~~~~~~ +\sum_{j \in \mathcal{N}_i}\bigg(z_{i \mid j}^{(t)} A_{i j} x_i+\frac{c}{2}\left\|A_{i j} x_i-\frac{1}{2} b_{i j}\right\|^2\bigg)\bigg),\\
    \nonumber&y_{i \mid j}^{(t+1)}=z_{i \mid j}^{(t)}+2 c\left(A_{i j} x_i^{(t+1)}-\frac{1}{2} b_{i j}\right), \\
    &z_{i\mid j}^{(t+1)}=\left \{ \begin{array}{ll}
        (1-\theta)z_{i\mid j}^{(t)}+\theta y_{j\mid i}^{(t+1)}, &y_{i\mid j}^{(t+1)}+y_{j\mid i}^{(t+1)}>0,\\ 
    (1-\theta)z_{i\mid j}^{(t)}-\theta y_{i\mid j}^{(t+1)}, &\text{otherwise,} \label{eq.z_update}
    \end{array}\right. 
\end{align}
where $y$ and $z$ are auxiliary variables, $\theta\in (0,1)$ is an avaraging constant and $c>0$ is a constant controling the convergence rate. When the objective function is uniformly convex, the algorithm will also converge for $\theta=1$ (standard PDMM) \cite{heusdens2024distributed}. Since the LP problem is not uniformly convex, we primarily focus on analyzing the case where $\theta=\frac{1}{2}$. The choice corresponds to the $\frac{1}{2}$-averaged version of PDMM, which is equivalent to ADMM.

\subsection{Adversary model and evaluation metrics}

\noindent\textbf{Adversary model}: We consider two widely used adversary models. The first is the passive adversary, represented by corrupt nodes in the graph. These nodes follow the algorithm's instructions but collude to gather and share information. We denote the set of corrupt nodes in the network by $\mathcal{V}_c$ and the set of honest nodes by $\mathcal{V}_h$. The second type is eavesdropping, which can intercept all messages transmitted through unencrypted channels. These two adversaries are assumed to be able to collaborate to infer the private data of honest nodes.

The performance of the algorithm is evaluated based on the following two requirements and their corresponding metrics.

\noindent\textbf{Output accuracy}: It measures how close the optimization results of the privacy-preserving algorithm are to those original non-privacy-preserving algorithms. We quantify the accuracy using the squared error $\|x_i^{(t_{\recht{max}})}- x^{*}\|_2^{2}$, where $t_{\recht{max}}$ denotes the maximum number of iterations and $x^*$ the optimal solution.

\noindent\textbf{Individual privacy}:
 Both $\epsilon$-DP and mutual information approaches are widely used information-theoretical methods for quantifying privacy \cite{cuff2016differential,Jane2020TIFS}. 
 We adopt mutual information as the metric for assessing individual privacy as it is shown effective in the literature \cite{duchi2013local,kairouz2014extremal,li2022privacy}. Given the random variable $S_i$ 
representing the private data at node $i$ and $\mathcal{O}$ representing the total information that the adversary can observe, the  mutual information $I(S_i;\mathcal{O})$ \cite{cover2012elements} measures the amount of information learned about $S_i$ by observing $\mathcal{O}$, which is give by
\begin{equation}
    \nonumber I(S_i;\mathcal{O})=h(S_i)-h(S_i\mid \mathcal{O}),
\end{equation}
where $h(\cdot)$ denotes differential entropy.
When $I(S_i;\mathcal{O})=h(S_i)$, the adversary has sufficient information to fully infer $s_i$. When $I(S_i;\mathcal{O})=0$, the adversary cannot infer any information about  $S_i$ given the available information ${\cal O}$.

\section{Proposed approach}

We now proceed to the proposed approach. We first introduce how to reformulate the problem using an augmented graph by adding dummy nodes to the network. One for each node, which serves the purpose of overcoming the limitation of requiring at least one honest neighbor for privacy preservation. That is, given node $i\in{\cal V}$, we introduce a dummy node $i'$. The new graph thus obtained will be denoted by $\mathcal{G}^\prime=(\mathcal{V}^\prime,\mathcal{E}^\prime)$ where $|{\cal V}'| = 2|{\cal V}|$. See Fig. \ref{fig:topo} for an illustration.
\begin{figure}[t]
\vspace{-5pt} 
    \centering
    \includegraphics[width=0.35\textwidth]{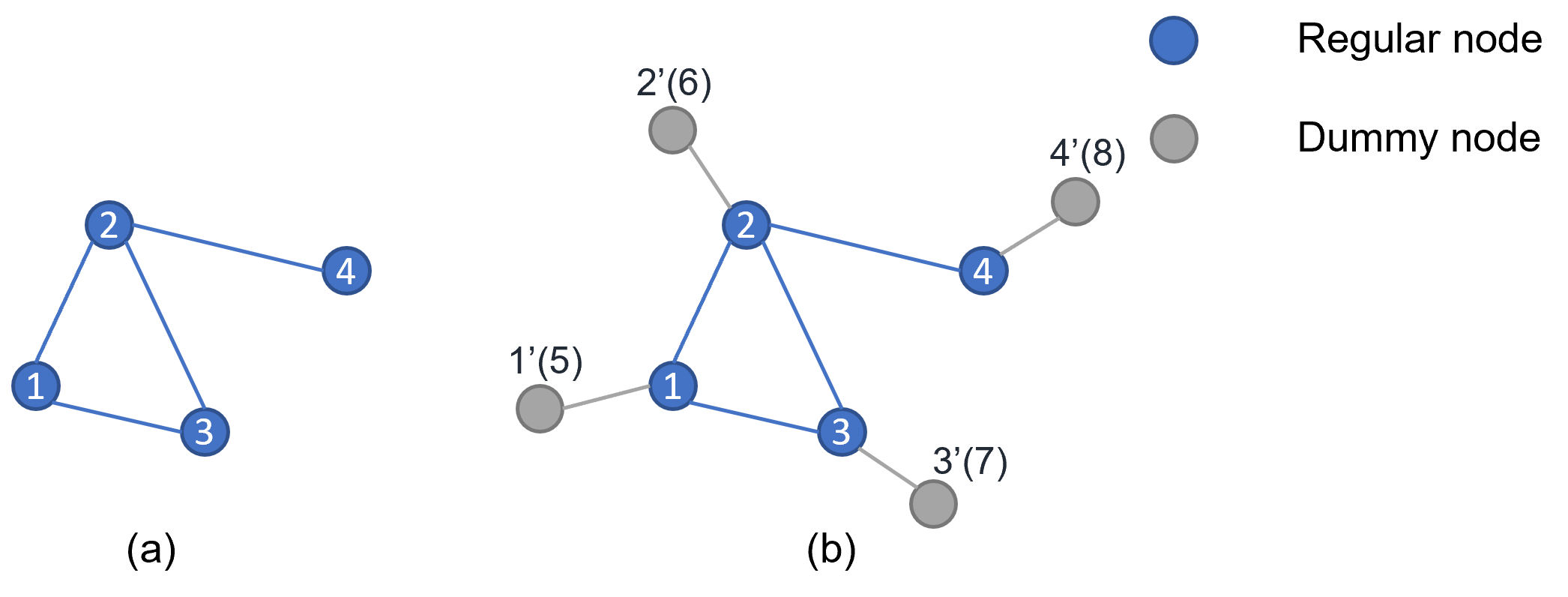}
    \vspace{-8pt} 
    \caption{(a) Example of the original graph $\mathcal{G}$; (b) Example of the augmented graph $\mathcal{G}^\prime$ including dummy nodes.}
    \label{fig:topo}
    \vspace{-12pt} 
\end{figure}
With this, we can formulate the constraints in \eqref{eq:lp} as
\begin{align*}
    A_{ij} x_i+A_{ji} x_j =0 \quad&\text{for} \quad(i,j)\in \cal E \\
    A_{ii'} x_i+A_{i'i} x_j \leq -s_i \quad&\text{for} \quad(i,i')\in \cal E^\prime\backslash \cal E
\end{align*}
where $A_{ij}= -A_{ji} = 1$ when $i<j$, and $A_{ii'}=-1$, $A_{i'i}=0$.

In addition to adding dummy nodes to the graph, we utilize the concept of subspace perturbation, initially introduced for distributed optimization with equality constraints \cite{li2020privacy}. The main idea is to properly initialize the auxiliary variable $\boldsymbol{z}^{(0)}$, thereby safeguarding the private data from being exposed without sacrificing the output accuracy.
Details of the proposed approach are summarized in Alg.~\ref{alg:pdmm}. 
Note that the updates at each node $i\in{\cal V}'$ can be done in parallel and that no direct collaboration is required between nodes during the computation of these updates, leading to an attractive (parallel) algorithm for optimization in practical networks.

\begin{algorithm}[t]
  \caption{Proposed approach}
  \label{alg:pdmm}
  \begin{algorithmic}
     \ForAll{ $i \in \mathcal{V}', j \in \mathcal{N}_i$,}
     \State Randomly initialize $z_{i\mid j}^{(0)}$ \Comment{Initialization}
     \State $\text{Node}_j \leftarrow \text{Node}_i(z_{  i|j}^{(0)})$
     \vspace{-5pt}
     \EndFor
          \For{$t=0,1,...$} 
            \ForAll{$i \in \mathcal{V}$ } 
                \State 
                \vspace{-8pt} 
                    \begin{equation} \label{eq:x_up}
                    x_{  i}^{(t+1)} =\frac{ -1- \sum_{  j \in {\cal N}_i}  A_{  ij}z_{  i|j}^{(t)}+(z^{(t)}_{i|i^\prime}+\frac{1}{2}cs_i)}{c(d_i+1)}
                \end{equation}
                \State \vspace{-3pt} 
                    \begin{equation}\label{eq:z_up}
                    \hspace{3mm}\forall j\in \mathcal{N}_i: z_{  j|i}^{(t+1)}
        =\frac{1}{2}z_{j\mid i}^{(t)}+\frac{1}{2}(z_{  i|j}^{(t)}+2c A_{  ij}x_i^{(t+1)})
                \end{equation}
                \State $\text{Node}_{j\in \mathcal{N}_i} \leftarrow \text{Node}_i(x_i^{(t+1)})$\Comment{Broadcast}
                \ForAll{ $j \in \mathcal{N}_i$} 
                \vspace{2pt}
                 \State $z_{i\mid j}^{(t+1)}$ from \eqref{eq:z_up} %
            \EndFor
            
        \State $y_{i|i^\prime}^{(t)}=z_{i|i^\prime}^{(t)}-2cx_i^{(t+1)}+cs_i$; $y_{i^\prime|i}^{(t)}=z_{i^\prime|i}^{(t)}+cs_i$
        \If {$y_{  i|i^\prime}^{(t)}+y_{i^\prime\mid i}^{(t)}> 0$}\Comment{Dummy updates}
        \State 
        \vspace{-8pt}
                    \begin{align}\label{eq:dummy_up1}
                    &z_{i\mid i^\prime}^{(t+1)}=\frac{1}{2}z_{i\mid i^\prime}^{(t)}+\frac{1}{2}
    y_{i^\prime|i}^{(t)}\\ \label{eq:dummy_up1p}
    &z_{i^\prime\mid i}^{(t+1)}=\frac{1}{2}z_{i^\prime\mid i}^{(t)}+\frac{1}{2}
    y_{i|i^\prime}^{(t)}
                \end{align}
                \vspace{-12pt} 
        \Else
        \State \vspace{-13pt} 
                    \begin{align}\label{eq:dummy_up2}
                    &z_{i\mid i^\prime}^{(t+1)}=\frac{1}{2}z_{i\mid i^\prime}^{(t)}-\frac{1}{2} y_{i|i^\prime}^{(t)}\\
    &z_{i^\prime\mid i}^{(t+1)}=\frac{1}{2}z_{i^\prime\mid i}^{(t)}-\frac{1}{2} y_{i^\prime|i}^{(t)}\label{eq:dummy_up2up}
                \end{align} 
        \EndIf
      \EndFor
      \EndFor
  \end{algorithmic}
  \vspace{-2pt} 
\end{algorithm}

We now analyze the performances of the proposed approach.
\vspace{-15pt} 
\subsection{Output accuracy} 
When subspace perturbation is applied to inequality-constrained problems, it is shown in \cite[Proposition 1]{heusdens2024distributed} that the optimization variable in A/PDMM, under both equality and inequality constraints, converge to the optimal solution, regardless of the initial values of the auxiliary variable. This ensures that the accuracy of the output is not compromised by the initialization choice of the auxiliary variable. Therefore, our primary focus will be on proving the privacy guarantees.

\vspace{-4pt} 
\subsection{Individual privacy}
Given that the eavesdropping adversary holds the information transmitted over all channels given by $\{x_{i}^{(t+1)}:t\geq 0, i\in \mathcal{V}\}\cup \{z_{i\mid j}^{(0)}: (i,j)\in\mathcal{E}\}$,  and  corrupt nodes hold local updates information $\{s_j,z_{j\mid i}^{(t)},z_{i\mid j}^{(t)}:t\geq 0, j\in \mathcal{V}_c, (i,j)\in\mathcal{E}'\}$.  
Let $\mathcal{T}=\{0,1,...,t_{\rm max}\}$. Given $i\in \mathcal{V}_h$, the individual privacy of honest node $i$ is defined as how much information about the private data $s_i$ can be inferred given the adversaries' knowledge. This is measured by
\begin{align} \label{eq.ind}
 I(S_i;\mathcal{O})=&I(S_i;\{S_j\}_{j\in\mathcal{V}_c},\{X_{j}^{(t+1)}\}_{j\in \mathcal{V},t\in\mathcal{T}},\\ \nonumber
    &\{Z_{j\mid k}^{(0)}\}_{ (j,k)\in\mathcal{E}},\{Z_{j\mid k}^{(t)},Z_{k\mid j}^{(t)}\}_{j\in \mathcal{V}_c,(j,k)\in \mathcal{E}^\prime,t\in\mathcal{T}})  
\end{align}

Without loss of generality, assuming the private data $s_i$s are drawn from independent distributions, our main result is given in Theorem~\ref{thm.1}, which states that the proposed approach can guarantee (asymptotically) perfect individual privacy even though all other nodes are corrupt, i.e., no information about its  private data $s_i$ can be inferred by the passive and eavesdropping adversaries.
\vspace{-2pt}
\begin{theorem}\label{thm.1}
Given $i\in{\cal V}_h$. If 
\begin{equation}
\forall t\in \mathcal{T}: ~ z_{  i|i^\prime}^{(t)}+z_{i^\prime\mid i}^{(t)}-2cx_i^{(t+1)}+2cs_i\leq 0,
\label{eq:cond}
\end{equation}
\vspace{-10pt}
\begin{align}\label{eq.perfect}
\textit{then}
\lim_{\sigma_{z}\rightarrow \infty }I(S_i;\mathcal{O})=\lim_{\sigma_{z}\rightarrow \infty }I(S_i;Z^{(0)}_{i|i^\prime}+\frac{1}{2}cS_i)\rightarrow 0,
\end{align}
\end{theorem}
\begin{proof}
\vspace{-3pt}
    See Appendix \ref{pf.thm1}.
    \vspace{-2pt}
\end{proof}
\noindent
Note that \eqref{eq:cond} is equivalent to the condition $y_{  i|i^\prime}^{(t)}+y_{i^\prime\mid i}^{(t)}\leq 0$, see \eqref{eq.z_update}. In other words, in order to preserve privacy we should avoid data exchange 
between dummy and regular nodes.

Several remarks are in place here. First, from \eqref{eq:cond}, it is clear that privacy is guaranteed by the honest node's dummy nodes, meaning no honest neighbor is required for privacy assurance. Second,  the node with maximum value will not satisfy condition \eqref{eq:cond} which is to be expected as we require perfect output accuracy so that the value $s_{\rm max}$ will be eventually available to all nodes. However, for the remaining nodes, the condition for privacy can be satisfied. In the following section, we will demonstrate that it is possible to meet this condition for all iterations by adjusting the convergence parameter $c$.

\section{Simulation results}
\label{sec:exp}

\textbf{Experimental setting:} We compare our method with existing privacy-preserving maximum consensus approaches \cite{venkategowda2020privacy,wang2018differentially}. We generate a random geometric graph (RGG) \cite{penrose2003random} with $n=10$ nodes. The private data, i.e. $s_i$ for $i \in \mathcal{V}$, are randomly drawn from a standard normal distribution $\mathcal{N}(0,1)$. The auxiliary variables $z_{i\mid i^\prime}^{(0)}$ and $z_{i^\prime\mid i}^{(0)}$ are drawn from $\mathcal{N}(\mu_z,\sigma_z^2)$ and $\mathcal{N}(-\mu_z,\sigma_z^2)$, respectively. Here $\mu_z$ can take a large value (we use $\mu_z=1000$ in the experiments) to ensure that condition \eqref{eq:cond} is satisfied; a larger value of $z_{i\mid i^\prime}$ will result in a correspondingly larger value of $x_i$. To counterbalance the influence of $z_{i\mid i^\prime}$ in \eqref{eq:cond}, however, we introduce a similar or larger negative value for $z_{i^\prime\mid i}$ at initialization.

\vspace{-4pt} 
\subsection{Information leakage via mutual information}
To visualize information loss in \eqref{eq.perfect} as a function of the variance of the inserted noise, we used NPEET toolbox \cite{npeet} to estimate the normalized mutual information $I(S_i;Z^{(0)}_{i|i^\prime}+\frac{1}{2}cS_i)/I(S_i;S_i)$ across $\sigma_z$, as depicted  Fig.~\ref{fig:nmi}. As expected, the information loss decreases notably as $\sigma_z$ increases.

\begin{figure}[htb]
\vspace{-8pt} 
    \centering
\includegraphics[width=0.25\textwidth]{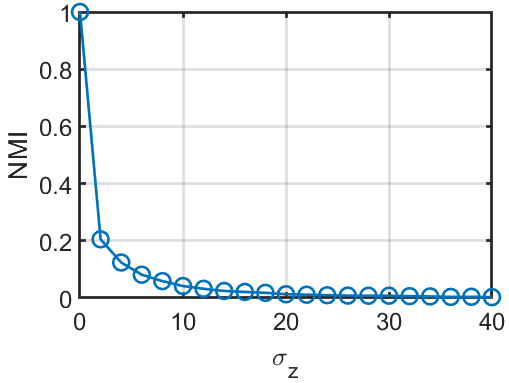}
\vspace{-16pt} 
    \caption{Normalized mutual information (NMI) $\frac{I(S_i;Z^{(0)}_{i|i^\prime}+\frac{1}{2}cS_i)}{I(S_i;S_i)}$ as a function of variance $\sigma_z$}.
    \label{fig:nmi}
    \vspace{-15pt} 
\end{figure}

\vspace{-4pt} 
\subsection{Performance comparison}
We first show that condition \eqref{eq:cond}, required in Theorem \ref{thm.1}, can be satisfied at all times by adjusting the convergence parameter $c$.  Fig.~\ref{fig:cond}  shows the LHS of \eqref{eq:cond} as a function of the iteration number for three different choices of the parameter $c$. We can see that 1) the blue curve, representing  $x_i^{(t)}$ of the node having the maximum value, does not meet condition \eqref{eq:cond}. This is expected as the maximum value will eventually be known to all nodes. 2) For other nodes, a larger parameter $c$ helps to satisfy condition \eqref{eq:cond}, thereby guaranteeing privacy. 

\begin{figure}[htb]
\vspace{-6pt} 
    \centering
\includegraphics[width=0.48\textwidth]{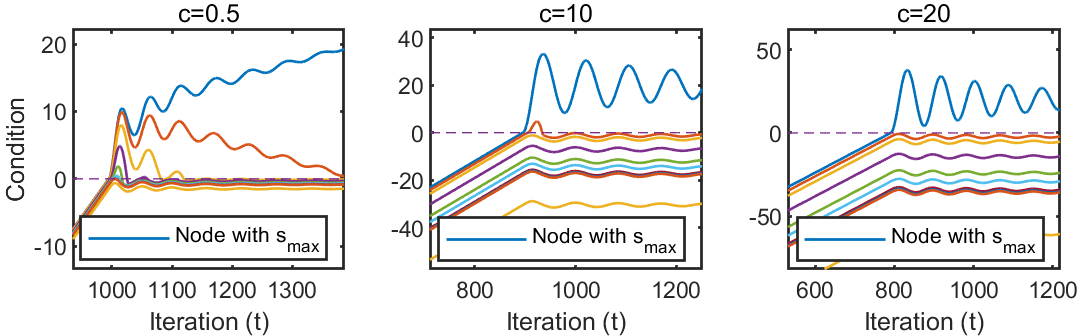}
\vspace{-6pt} 
    \caption{LHS of \eqref{eq:cond} as a function of $t$ for three values of  $c$, where the blue lines are the results for the node having the maximum value and the others for nodes having $s_i < s_{\rm max}$.}
    \label{fig:cond}
    \vspace{-2pt} 
\end{figure}

In Fig.~\ref{fig:dummy} we compare our proposed approach with two existing algorithms \cite{wang2018differentially, venkategowda2020privacy} using three privacy levels $\sigma=10^{-2},10^{-1},10^{0}$, respectively.
It is evident that as the noise increases, both existing approaches exhibit a pronounced deterioration in accuracy, highlighting the trade-off between privacy and accuracy. In contrast, our proposed method converges to the optimal result regardless of the noise variance, demonstrating that it does not compromise accuracy for privacy. This is further detailed in Fig. \ref{fig:x} where the convergences of minimum, median and maximum nodes are illustrated. 

\begin{figure}[htb]
\vspace{-12pt} 
    \centering
\includegraphics[width=0.47\textwidth]{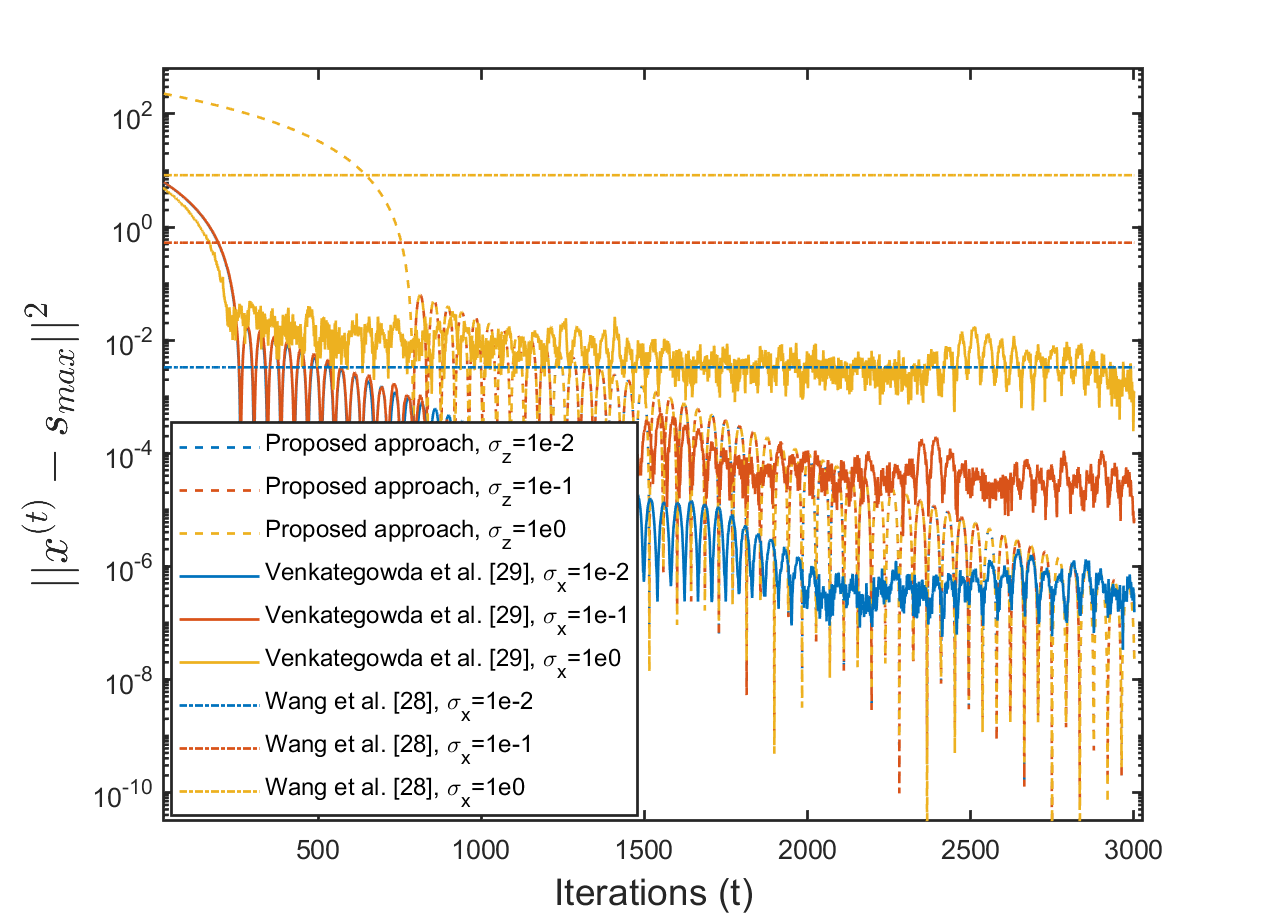}
\vspace{-7pt} 
    \caption{Performance comparison of the proposed approach with two existing approaches under various privacy levels.}
    \label{fig:dummy}
    \vspace{-6pt} 
\end{figure}

\begin{figure}[htb]
\vspace{-6pt} 
    \centering
\includegraphics[width=0.49\textwidth]{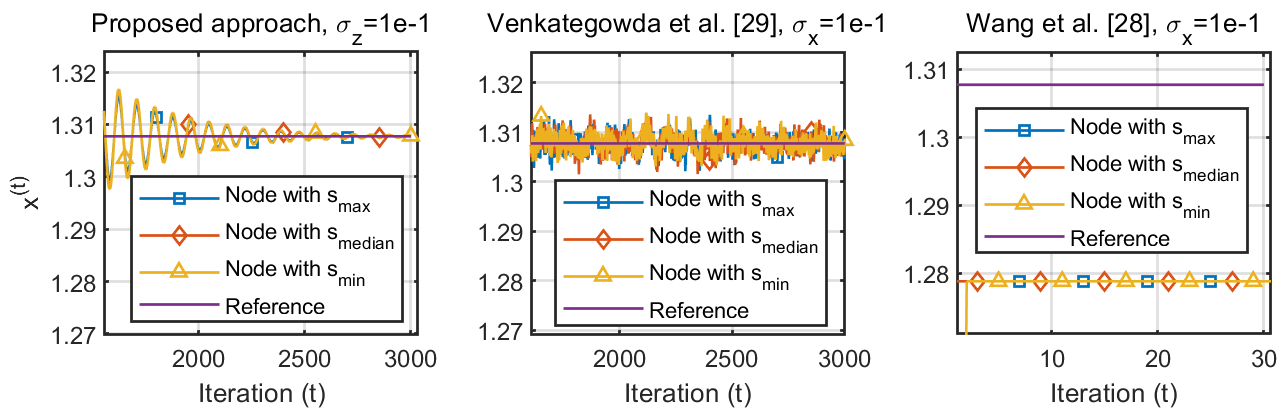}
\vspace{-6pt} 
    \caption{Convergence of the optimization variable $x^{(t)}$ for three nodes with minimum, median and maximum value of three algorithms, respectively.}
    \label{fig:x}
    \vspace{-12pt} 
\end{figure}

\section{Conclusion}
In this paper, we proposed a novel privacy-preserving distributed maximum consensus algorithm. Our method involves adding dummy nodes to form an augmented graph and applying inequality constraint-based subspace perturbation, ensuring the privacy of honest participants. Using information-theoretical measures, we demonstrated that the proposed approach can guarantee perfect privacy against both eavesdropping and passive adversaries. Furthermore, the method preserves privacy without compromising accuracy, maintaining superior performance. Experimental results further consolidate the superiority of our approach compared to existing methods.
\appendix
\section{Appendix}
\subsection{Proof of Theorem \ref{thm.1}} \label{pf.thm1}
Replacing $z_{j\mid i}^{(t)}$ and $z_{i\mid j}^{(t)}$  in \eqref{eq:z_up} we obtain  
\begin{align} \label{eq:zdif}  
z_{j\mid i}^{(t+1)}-z_{j\mid i}^{(t)}=cA_{ij}x_i^{(t+1)}-\frac{1}{2}cA_{ij}x_i^{(t)}+\frac{1}{2}cA_{ji}x_j^{(t)}.
\end{align}
Moreover, considering the difference $x_j^{(t+2)}-x_j^{(t+1)}$ using \eqref{eq:x_up} and combining with \eqref{eq:zdif} we obtain
\begin{footnotesize}
\begin{align} \label{eq:zxSuf}
    &x_j^{(t+2)}-x_j^{(t+1)} \nonumber\\
    &=\frac{c\sum_{  k \in {\cal N}_j}  (x_k^{(t+1)}-\frac{1}{2}x_k^{(t)}-\frac{1}{2}x_j^{(t)})+(z^{(t+1)}_{j|j^\prime}-z^{(t)}_{j|j^\prime})}{c(d_j+1)}.
\end{align}
\end{footnotesize}
With this, \eqref{eq.ind} becomes 
\begin{footnotesize}
    \begin{align*}
    &I(S_i;\mathcal{O}) \\
    \overset{(a)}{=}&I(S_i;\{S_j,Z_{j\mid j^\prime}^{(0)},Z_{j^\prime\mid j}^{(0)}\}_{j\in\mathcal{V}_c},\{X_{j}^{(t+1)}\}_{j\in \mathcal{V},t\in\mathcal{T}},\{Z_{j\mid k}^{(0)}\}_{ (j,k)\in\mathcal{E}})
    \\
    \overset{(b)}{=}&I(S_i;\{S_j,Z_{j\mid j^\prime}^{(0)},Z_{j^\prime\mid j}^{(0)}\}_{j\in\mathcal{V}_c},\{Z_{j\mid j^\prime}^{(t+2)}-Z_{j\mid j^\prime}^{(t+1)}\}_{j\in \mathcal{V},t\in\mathcal{T}},\\ 
    &\{X_{j}^{(1)},X_{j}^{(2)}\}_{j\in \mathcal{V}},\{Z_{j\mid k}^{(0)}\}_{ (j,k)\in\mathcal{E}})\\
    \overset{(c)}{=}&I(S_i;\{S_j,Z_{j\mid j^\prime}^{(0)},Z_{j^\prime\mid j}^{(0)}\}_{j\in\mathcal{V}_c},\{X_{j}^{(1)},X_{j}^{(2)}\}_{j\in \mathcal{V}},\{Z_{j\mid k}^{(0)}\}_{ (j,k)\in\mathcal{E}})
    \\
    \overset{(d)}{=}&I(S_i;\{S_j,Z_{j\mid j^\prime}^{(0)},Z_{j^\prime\mid j}^{(0)}\}_{j\in\mathcal{V}_c},\{Z_{j\mid k}^{(0)}\}_{(j,k)\in\mathcal{E}},\{Z^{(0)}_{j|j^\prime}+\frac{1}{2}cS_j\}_{j\in \mathcal{V}_h})
        \\
    \overset{(e)}{=}&I(S_i;Z^{(0)}_{i|i^\prime}+\frac{1}{2}cS_i)
\end{align*}
\end{footnotesize}
where (a) follows from \eqref{eq:z_up}, \eqref{eq:dummy_up1} and \eqref{eq:dummy_up2}  since the $z$ variables can be derived from previous $z,x_i$, and $s_i$ values, (b) follows from \eqref{eq:zxSuf} since $\{x_j^{(1)},x_j^{(2)}\}_{j\in \mathcal{V}}$ and  $\{z_{j\mid j^\prime}^{(t+2)}-z_{j\mid j^\prime}^{(t+1)}\}_{j\in \mathcal{V},t\in\mathcal{T}}$  are sufficient to compute all $\{x_j^{(t+1))}\}_{j\in \mathcal{V},t\in\mathcal{T}}$, and vise versa,   (c) assumes that  condition \eqref{eq:cond} is satisfied so that $ z_{j\mid j^\prime}^{(t+2)}-z_{j\mid j^\prime}^{(t+1)}=c(x_i^{(t+2)}-x_i^{(t+1)})$ can be recursively computed from $\{x_j^{(1)},x_j^{(2)}\}_{j\in \mathcal{V}}$, (d) holds since $\{x_{j}^{(1)},x_{j}^{(2)}\}_{j\in \mathcal{V}_c}$ can be computed from $z^{(0)}$ and $s$ from corrupt nodes, while (e) holds because $S_j, j\neq i,$ and all $z$s are independent of $S_i$.
Hence, when $\sigma_{Z_{i\mid i^\prime}}\rightarrow\infty$, $Z_{j\mid j\prime}^{(0)}+\frac{1}{2}cS_j$ becomes independent of $S_i$, and thus $I(S_i;\mathcal{O})\rightarrow 0$, thereby completing the proof.

\newpage
\bibliographystyle{IEEEbib}
\bibliography{refs}
\end{sloppy}
\end{document}